\documentclass[letterpaper, 10 pt, conference]{ieeeconf}

\IEEEoverridecommandlockouts                              

\usepackage[noadjust]{cite}
\usepackage{graphics} 
\usepackage{epsfig} 
\usepackage{times} 
\usepackage{amssymb, amsmath}
\usepackage{bm}
\usepackage{dsfont}
\usepackage[font=small,labelfont=bf]{caption}
\usepackage[labelformat=simple]{subcaption}

\usepackage{breqn}
\usepackage{array}
\usepackage{float}
\usepackage{siunitx}
\usepackage{mathtools}
\usepackage{float}
\usepackage{siunitx}

\usepackage{amsthm}
\usepackage{xspace}
\usepackage[linesnumbered,ruled,vlined]{algorithm2e}
\usepackage{array}
\usepackage{algorithmic}
\usepackage{makecell}
\usepackage{booktabs}

\usepackage{threeparttable}
\usepackage{xcolor}
\usepackage{hyperref}

\newtheorem{definition}{Definition}

\newtheorem{assumption}{Assumption}
\newtheorem{proposition}{Proposition}

\DeclareMathOperator{\MSE}{MSE}
\DeclareMathOperator{\cayley}{Cayley}
\DeclareMathOperator{\Lip}{Lip}
\DeclareMathOperator{\diag}{diag}

\DeclareMathOperator{\sign}{sign}

\renewcommand{\cal}[1]{\mathcal{ #1 }}

\newcommand{\bb}[1]{\mathbb{ #1 }}

\newcommand{\grad}{\nabla}

\newcommand{\union}{\cup}
\newcommand{\R}{\bb{R}}

\DeclarePairedDelimiterX{\Set}[2]\{\}{%
  \, #1 \;\delimsize\vert\; #2 \,
}

\title{\LARGE \bf
Data-Efficient System Identification via Lipschitz Neural Networks
}

\author{Shiqing Wei, Prashanth Krishnamurthy, and Farshad Khorrami
\thanks{The authors are with Control/Robotics Research Laboratory, Department of ECE, NYU Tandon School of Engineering, 5 MetroTech Center, Brooklyn, NY 11201, USA. {\tt\small \{shiqing.wei, prashanth.krishnamurthy, khorrami\}@nyu.edu}}%
\thanks{This work was supported in part by ARO grant W911NF-21-1-0155 and by the New York University Abu Dhabi (NYUAD) Center for Artificial Intelligence and Robotics (CAIR), funded by Tamkeen under the NYUAD Research Institute Award CG010.}
}

\begin{document}
\maketitle
\thispagestyle{empty}
\pagestyle{empty}

\begin{abstract}
Extracting dynamic models from data is of enormous importance in understanding the properties of unknown systems. In this work, we employ Lipschitz neural networks, a class of neural networks with a prescribed upper bound on their Lipschitz constant, to address the problem of data-efficient nonlinear system identification. Under the (fairly weak) assumption that the unknown system is Lipschitz continuous, we propose a method to estimate the approximation error bound of the trained network and the bound on the difference between the simulated trajectories by the trained models and the true system. Empirical results show that our method outperforms classic fully connected neural networks and Lipschitz regularized networks through simulation studies on three dynamical systems, and the advantage of our method is more noticeable when less data is used for training.

\end{abstract}

\section{Introduction}
Dynamic models play a crucial role in understanding and forecasting behaviors across various research fields, such as biology, physics, and engineering \cite{pillonetto2014kernel, brunton2016discovering}. Despite the extensive application of dynamic models, the governing equations are typically derived under ideal conditions, not only requiring specific domain knowledge but also having a non-negligible mismatch with respect to the real-world problem. The alternative/inverse process of extracting mathematical models from observed data is commonly referred to as system identification \cite{ljung1998system}. In this work, we leverage recent advances in machine learning and propose a novel system identification method using neural networks.

More specifically, we address system identification for nonlinear dynamical systems via Lipschitz neural networks, a class of neural networks with a prescribed Lipschitz bound. Under the (fairly weak) assumption that the underlying unknown system is Lipschitz continuous, we propose a bound on the approximation error of our method on a given state space and characterize the difference between the simulated solution by our method and the solution to the true system. Empirical results show that soft regularization can result in a high variance of the Lipschitz bound of the trained network, and the benefits of our method are more evident when less training data is used. 

A variety of approaches have been proposed in the field of system identification. Sparse regression-based methods find sparse coefficients among a user-defined functional class to find a good fit to the measurement data and have been used to identify both ODEs (\cite{brunton2016discovering, egan2024automatically}) and PDEs (\cite{rudy2017data}). Koopmanism-based approaches (e.g., \cite{nathan2018applied}) convert the system identification problem to a linear identification problem in the space of the observables (which is usually finite-dimensional for practical implementation). Additionally, many kernel-based techniques have also been proposed for system identification in the linear case (see \cite{pillonetto2014kernel} for a comprehensive review of the kernel methods on this topic). Typically, domain knowledge and experience play an important part in the aforementioned methods, as the performance of these methods is dependent on the functional class (or kernels) provided by the user.

With the availability of powerful computation resources, deep learning becomes an emergent approach to system identification. The authors of \cite{rahman2022neural} employ neural ordinary differential equations (NODEs) for system identification, and the training is conducted to reduce the mismatch between the system trajectory and the trajectory simulated by NODEs. A particularly relevant line of work, \cite{negrini2021system, negrini2023robust}, proposes Lipschitz regularized networks (LRNs) for system identification, where an additional term penalizing the Lipschitz constant of the networks is added to the training loss. In this work, instead of penalizing the Lipschitz constant, we employ Lipschitz neural networks that come with a natural regularization through the bound on its Lipschitz constant. The Lipschitz continuity of such neural networks results from a specific parameterization of the trainable parameters (e.g., \cite{wang2023direct, havens2024exploiting}), and ongoing research shows that they are as accurate as classic neural networks in classification tasks \cite{bethune2022pay}.

\textit{Our contributions:} (1) We design and develop a novel learning framework that uses Lipschitz neural networks for nonlinear system identification. (2) We propose a method to bound the approximation error of our method (when the unknown system also has Lipschitz continuous dynamics) and a bound on the difference between the simulated solution by our method and the solution to the true system. (3) We demonstrate the effectiveness of our method through simulations on a linear system, the Van der Pol oscillator, and a two-link planar robot arm. Comparison studies show that our method outperforms classic fully connected neural networks and the LRNs, and the benefits of our method are more pronounced when less data is used. 

\textit{Notations:} Let $\R$ be the set of real numbers, $\mathbb{D}_{++}^n$ be the set of $n \times n$ positive definite diagonal matrices, and $I$ be the identity matrix of proper dimensions. $\lVert \cdot \rVert_p$ is the $\ell_p$ norm of a vector or the induced norm of a matrix. $\diag(\cdot)$ gives a diagonal matrix from a vector $v$. $|\cal{S}|$ represents the cardinality of a set $\cal{S}$. $B_p(x_0, r) = \{ x \in \R^n \; |\; \lVert x-x_0 \rVert_p \leq r \}$ is the $\ell_p$ (with $p = 1, ..., \infty$) ball centered at $x_0 \in \R^n$ with $r > 0$.

\section{Preliminaries}
In this section, we introduce a class of neural network layers with prescribed Lipschitz properties. 
\begin{definition}\label{def:lip_continuous}
    A function $\phi: \R^n \rightarrow \R^m$ is Lipschitz continuous (or $\gamma$-Lipschitz), if for all $x_1, x_2 \in \R^n$, there exists a positive real constant $\gamma$ such that
    \begin{equation}\label{eq:def_lip}
        \lVert \phi(x_1)-\phi(x_2) \rVert_2 \leq \gamma \lVert x_1 - x_2 \rVert_2.
    \end{equation}
    The smallest constant such that \eqref{eq:def_lip} holds is called the Lipschitz constant of $\phi$, which is denoted by $\Lip(\phi)$.
\end{definition}

Many neural networks can be seen as a sequential composition of different layers. Denote by $h_{\mathrm{in}} \in \R^{n_{\mathrm{in}}}$ and $h_{\mathrm{out}} \in \R^{n_{\mathrm{out}}}$ the input and output of a certain layer, respectively. The following 1-Lipschitz layer was proposed in \cite{wang2023direct}.
\begin{proposition}[Theorem 3.2 \cite{wang2023direct}]
    Let $\Psi \in \mathbb{D}_{++}^{n_{\mathrm{out}}}$, $A \in \R^{n_{\mathrm{out}} \times n_{\mathrm{out}}}$, $B \in \R^{n_{\mathrm{out}} \times n_{\mathrm{in}}}$, and $b \in \R^{n_{\mathrm{out}}}$. If $A A^\top + B B^\top = I$ and $\sigma$ is an activation function\footnote{$\sigma$ is applied element-wise.} with slope restricted in $[0,1]$, then
    \begin{equation}\label{eq:one_lip_layer}
        h_{\mathrm{out}} = \sqrt{2} A^\top \Psi \sigma \left( \sqrt{2} \Psi^{-1} B h_{\mathrm{in}} + b \right)
    \end{equation}
    is a 1-Lipschitz layer.
\end{proposition}

The weight matrices $A$ and $B$ in \eqref{eq:one_lip_layer} can be obtained from the Cayley transform (see \cite{trockman2021orthogonalizing}). More specifically, for any matrices $X \in \R^{n_{\mathrm{out}} \times n_{\mathrm{out}}}$ and $Y \in \R^{n_{\mathrm{in}} \times n_{\mathrm{out}}}$, let
\begin{equation}\label{eq:cayley_transform}
    \begin{bmatrix}
        A^\top \\ B^\top 
    \end{bmatrix} = 
    \cayley \left( \begin{bmatrix}
        X^\top \\ Y^\top 
    \end{bmatrix} \right) = 
    \begin{bmatrix}
        (I+Z)^{-1} (I-Z) \\
        -2Y (I+Z)^{-1}
    \end{bmatrix}
\end{equation}
where $Z = X-X^\top + Y^\top Y$. As $I+Z$ is nonsingular\footnote{This is because $Z$ is the sum of a skew-symmetric matrix and a positive semidefinite matrix, and therefore is positive semidefinite.} and $(I-T)(I+T)^{-1} = (I+T)^{-1}(I-T)$ holds for a square matrix $T$ if $I+T$ is nonsingular (therefore this holds for $T=Z$ and $T = Z^\top$), we can verify that $A A^\top + B B^\top = I$. As for the matrix $\Psi$ in \eqref{eq:one_lip_layer}, it can be constructed by 
\begin{equation}
    \Psi = \diag([e^{v_1}, e^{v_2}, \ldots, e^{v_{n_{\mathrm{out}}}} ])
\end{equation}
where $v$ can be any vector in $\R^{n_{\mathrm{out}}}$.

\section{System Identification using Lipschitz Neural Networks} \label{sec:sys_id_through_lip_nn}
\subsection{Lipschitz Neural Networks}
Lipschitz neural networks are a class of neural networks with a prescribed upper bound of the Lipschitz constant. Let $n_i$ be the output dimension of the $i$-th layer. Define the following network $\Phi: \R^n \rightarrow \R^{n_L}$ with $L$ layers:
\begin{subequations}\label{eq:nn}
    \begin{align}
        \phi_{L-1} (x) &= h_{L-1} \circ h_{L-2} \circ \cdots \circ h_1 \circ F (x) \\
        \phi_L (x) &= \gamma' B_L \phi_{L-1} (x) \label{eq:nn_phi} \\
        \Phi (x) &= \phi_L (x) - \phi_L (0), \label{eq:zero_at_zero}
    \end{align}
\end{subequations}
where $\gamma' > 0$ is a design parameter, $h_1, h_2, \ldots, h_{L-1}$ are 1-Lipschitz layers as in \eqref{eq:one_lip_layer}, $B_L \in \R^{n_L \times n_{L-1} }$ such that $\lVert B_L \rVert_2 \leq 1$, and $F: \R^n \rightarrow \R^n$ is an affine function
\begin{equation}\label{eq:affine_func}
    F(x) = A_F (x - b_F)
\end{equation}
with $b_F \in \R^n$ and $A_F \in \R^{n \times n}$. We see that $\Phi(0) = 0$ by the design in \eqref{eq:zero_at_zero}. The affine function $F$ is chosen to perform a linear transformation to center and normalize the input data (and then fixed during training), which is a common practice to accelerate and stabilize the training of neural networks \cite{lecun2002efficient}. To obtain the matrix $B_L$ with $\lVert B_L \rVert_2 \leq 1$, we employ again the Cayley transform in \eqref{eq:cayley_transform} to obtain a pair of matrices $(A_L, B_L)$ with $A_L A_L^\top + B_L B_L^\top = I$, and keep only the matrix $B_L$. The next result shows that $\Phi$ has a prescribed bound on its Lipschitz constant. 

\begin{proposition} \label{prop:lip_bound}
    The neural network $\Phi$ in \eqref{eq:nn} is $\gamma$-Lipschitz with $\gamma = \gamma' \lVert A_F \rVert_2$.
\end{proposition}
\begin{proof}
    We first prove that $\phi_L$ in \eqref{eq:nn_phi} is $\gamma$-Lipschitz. Noting that $\phi_L$ is the composition of Lipschitz continuous functions $F, h_1, \ldots, h_{L-1}$, and $x \mapsto \gamma' B_L x + b_L$, we thus have 
    \begin{equation*}
        \Lip (\phi_L) \leq \gamma' \lVert A_F \rVert_2 \lVert B_L \rVert_2 \prod_{i=1}^{L-1} \Lip (h_i) \leq \gamma' \lVert A_F \rVert_2 = \gamma
    \end{equation*}
    as $\lVert B_L \rVert_2 \leq 1$ and $h_1, \ldots, h_{L-1}$ are all 1-Lipschitz. To see that $\Phi$ is also $\gamma$-Lipschitz, take $x, y \in \R^n$, and we have $\lVert \Phi(x) - \Phi(y) \rVert_2 = \lVert \phi_L(x) - \phi_L(y) \rVert_2 \leq \gamma \lVert x - y \rVert_2.$
\end{proof}

\subsection{System Identification}\label{sec:sys_id}
Our objective is to identify the following nonlinear system
\begin{equation}\label{eq:dynamics}
    \dot{x} = f(x)
\end{equation}
where $x \in \cal{X}$ is the state, $\cal{X} \subset \R^n$ is the state space (a subset of $\R^n$ to model state constraints), and $f: \cal{X} \rightarrow \R^n$ is $K$-Lipschitz continuous (in the $\ell_2$ norm) on $\cal{X}$, guaranteeing the existence and uniqueness of the solution on $\cal{X}$. Without loss of generality, we assume that $f(0) = 0$. If this is not the case, we can define the system using the new state variable $x' = x-x_e$ where $x_e$ is an equilibrium point of the system. Typically, $f$ is unknown or partially known. The task is to approximate $f$ numerically from data using Lipschitz neural networks, i.e.,  $\Phi$ is the estimate generated for $f$.

We assume that full-state measurements are available and collect the trajectories of $x(t)$ sampled at discrete time steps starting from different initial conditions. This is a common setting adopted in many studies in system identification (e.g., see \cite{brunton2016discovering, egan2024automatically, negrini2021system}). Realistically, $\dot{x}(t)$ is usually not available for direct measurement, and its approximated value $\widehat{\dot{x}}(t)$ is obtained by numerical differentiation of $x(t)$.

Let the collected dataset be $\cal{D} = \{ (x_i, y_i)  \; | \; x_i = x(t_i),\ y_i = \widehat{\dot{x}}(t_i),\ i = 1, \ldots, N\}$ where $t_i$ are the sampled timestamps. Then, we calculate the sample mean and variance of all $x(t_i)$ to obtain $b_F$ and $A_F$ in \eqref{eq:affine_func} (to center and normalize the input data). The dataset is randomly split into a training set $\cal{D}_{\mathrm{train}}$ and a test set $\cal{D}_{\mathrm{test}}$. The mean squared error (MSE) has been a classic criterion in evaluating the model's accuracy in many works on system identification (e.g., \cite{negrini2021system, roll2005nonlinear}). Define the MSE on a set $\cal{S}$ as 
\begin{equation}\label{eq:mse}
    \MSE (\cal{S}, \Phi) = \frac{1}{|\cal{S}|} \sum_{(x_i, y_i) \in \cal{S}} \lVert y_i - \Phi(x_i) \rVert_2^2.
\end{equation}
We train the neural network $\Phi$ by minimizing the MSE on the batches (subsets) of $\cal{D}_{\mathrm{train}}$. The training process is summarized in Algorithm \ref{algo:sys_id}.

\begin{algorithm}[t]
\SetAlgoLined
\KwIn{Neural network $\Phi$, dataset $\cal{D}$, learning rate scheduler StepLR, and number of epochs $N$.} 

Initialize and fix the values of $A_F$ and $b_F$ in \eqref{eq:affine_func} based on the sample mean and variance of $\cal{D}$\;

Randomly split $\cal{D}$ into $\cal{D}_{\mathrm{train}}$ and $\cal{D}_{\mathrm{test}}$\;

Further divide $\cal{D}_{\mathrm{train}}$ into batches $\cal{B}_1, \ldots, \cal{B}_M$\;

\For{$i = 1, 2, ..., N$}{
    Determine the learning rate $\alpha_i$ from StepLR\;
    
    \For{$j = 1, 2, ..., M$}{
        Calculate $\cal{L} = \MSE(\cal{B}_j, \Phi)$\ and $g = \grad_{\theta} \cal{L}$ where $\theta$ are the trainable parameters of $\Phi$\;

        Clip the gradient by $g \leftarrow g/\lVert g \rVert_2$ if $\lVert g \rVert_2 > 1$\;

        Update $\theta$ by $\theta \leftarrow \theta - \alpha_i g$\;
    }

    Calculate $\cal{L}_{\mathrm{test}} = \MSE(\cal{D}_{\mathrm{test}}, \Phi)$ and keep the $\theta^\star$ that gives the minimal $\cal{L}_{\mathrm{test}}$ during training\;
}
\caption{System id. via Lipschitz networks}
\label{algo:sys_id}
\end{algorithm}

\subsection{Verification of the Lipschitz Neural Networks}\label{sec:eval}
In this section, we present the theoretical analysis of verifying Lipschitz neural networks, e.g., provable bounds on the deviation between the learned and actual models and on the deviation between the trajectories of the estimated and actual dynamics.

The test set $\cal{D}_{\mathrm{test}}$ can be seen as a collection of independent and identically distributed (i.i.d) samples from the true data distribution $\rho$. Then, Hoeffding inequality gives us a bound on the difference between the actual MSE and the empirical one using a probabilistic description \cite[Section 10.1]{tempo2013randomized}. In the following, we seek to evaluate the accuracy of $\Phi$ from a deterministic perspective. 

\begin{assumption}\label{ass:bound_on_est_der}
    For all $(x_i, y_i) \in \cal{D}$, there exists a constant $c > 0$ such that $\lVert y_i - f(x_i) \rVert_2 \leq c $.
\end{assumption}

The next result introduces a bound on the maximum deviation of $\Phi$ from $f$ on a compact set.

\begin{proposition}\label{prop:max_error_bound}
Let $e_i = \Phi(x_i) - y_i$ for all $(x_i, y_i) \in \cal{D}$,  $\cal{X} \in \R^n$ be compact, and $C = \{ B_p (x_i, r_i) \}_{(x_i, y_i) \in \cal{D}}$ be a cover of $\cal{X}$, i.e., $\cal{X} \subseteq \cup_{(x_i, y_i) \in \cal{D}} B_p (x_i, r_i)$. If Assumption \ref{ass:bound_on_est_der} holds, then $\max_{x \in \cal{X}} \lVert \Phi(x) - f(x) \rVert_2 \leq c + \max_{(x_i, y_i) \in \cal{D}} \left[ n^{\max(0,\frac{1}{2} - \frac{1}{p})} (K+\gamma)r_i + \lVert e_i \rVert_2 \right]$ where $K$ is the Lipschitz bound of $f$.
\end{proposition}
\begin{proof}
    Take $(x_i, y_i) \in \cal{D}$. For all $x \in B_p (x_i, r_i)$, we have $\lVert x - x_i \rVert_2 \leq n^{\max(0,\frac{1}{2} - \frac{1}{p})} r_i$ by norm equivalence. Then, $\lVert \Phi(x) - f(x) \rVert_2 \leq n^{\max(0,\frac{1}{2} - \frac{1}{p})} (K+\gamma)r_i + \lVert \Phi(x_i) - f(x_i) \rVert_2$ as $\Lip(\Phi - f) \leq K+\gamma$. By Assumption \ref{ass:bound_on_est_der}, we further have $\lVert \Phi(x_i) - f(x_i) \rVert_2 \leq \lVert \Phi(x_i) - y_i \rVert_2 + \lVert y_i - f(x_i) \rVert_2 \leq \lVert e_i \rVert_2 + c$.
    Then, for all $x \in B_p (x_i, r_i)$, $\lVert \Phi(x) - f(x) \rVert_2 \leq n^{\max(0,\frac{1}{2} - \frac{1}{p})} (K+\gamma)r_i + \lVert e_i \rVert_2 + c$. As $C$ is a cover of $\cal{X}$, we have $\max_{x \in \cal{X}} \lVert \Phi(x) - f(x) \rVert_2 \leq c + \max_{(x_i, y_i) \in \cal{D}} \left[ n^{\max(0,\frac{1}{2} - \frac{1}{p})} (K+\gamma)r_i + \lVert e_i \rVert_2 \right] $.
\end{proof} 

\setlength\tabcolsep{2.2 pt}
\begin{table*}[t]
    \centering
    \caption{Lipschitz bounds and estimation error bounds of FCNs, LRNs, and our method. The results include the mean and standard deviation for trained networks with the best test MSEs, using 100\% of the training data across four random seeds (0, 100, 200, and 300).}
    \label{tab:max_error}
    \begin{tabular}{cccccccccc}
        \toprule
         &  \multicolumn{3}{c}{Linear system}&  \multicolumn{3}{c}{Van der Pol Oscillator}&  \multicolumn{3}{c}{Two-link robotic arm}\\
         & Lip. bound & $\delta = 0.05$ & $\delta = 0.025$ & Lip. bound & $\delta = 0.05$ & $\delta = 0.025$ & Lip. bound & $\delta = 0.05$ & $\delta = 0.025$ \\
         \midrule
         FCNs&  $18.91 \pm 5.08$ &  $2.79 \pm 0.67$ & $2.36 \pm 0.56$ &  $5.27 \pm 0.19$ &  $1.12 \pm 0.03$ & $1.10 \pm 0.03$ &  $15.18 \pm 5.06$ &  $4.70 \pm 1.50$& $4.11 \pm 1.32$\\
         LRNs&  $14.19 \pm 2.56$ &  $2.19 \pm 0.33$ & $1.85 \pm 0.28$ &  $17.16 \pm 1.62$ &  $3.00 \pm 0.26$ & $2.94 \pm 0.25$ &  $10.89 \pm 0.63$ &  $3.43 \pm 0.19$ & $3.00 \pm 0.16$\\
         Ours&  $\bm{2.01}$ & $\bm{0.64 {\scriptstyle \pm 0.0009}}$ & $\bm{0.51 {\scriptstyle \pm 0.0014}}$ & $\bm{4.02}$ & $\bm{0.92 {\scriptstyle \pm 0.0004}}$ & $\bm{0.90 {\scriptstyle \pm 0.0005}}$ &  $\bm{2.55}$ & $\bm{0.95 {\scriptstyle \pm 0.0016}}$ & $\bm{0.83 {\scriptstyle \pm 0.0013}}$ \\
         \bottomrule
    \end{tabular}
\end{table*}

\setlength\tabcolsep{2.2 pt}
\begin{table*}[t]
    \centering
    \caption{Mean and standard deviation of the test MSE of FCNs, LRNs, and our method trained on 25\%, 50\%, 100\% of the training data. Training is repeated using four random seeds (0, 100, 200, and 300).}
    \label{tab:test_mse}
    \begin{tabular}{cccccccccc}
        \toprule
         &  \multicolumn{3}{c}{Linear system ($10^{-3}$ units)}&  \multicolumn{3}{c}{Van der Pol Oscillator ($10^{-3}$ units)}&  \multicolumn{3}{c}{Two-link robotic arm ($10^{-3}$ units)}\\
         $\%$ & $25\%$ & $50\%$ & $100\%$ & $25\%$ & $50\%$ & $100\%$ & $25\%$ & $50\%$ & $100\%$ \\
         \midrule
         FCNs&  $5.47 \pm 0.04$ &  $5.41 \pm 0.03$ & $5.41 \pm 0.03$ &  $2.71 \pm 0.01$ &  $2.69 \pm 0.01$ & $2.68 \pm 0.01$ &  $2.71 \pm 0.01$ &  $2.70 \pm 0.01$& $\bm{2.68 \pm 0.01}$\\
         LRNs&  $8.03 \pm 0.5$ &  $6.28 \pm 0.08$ & $5.77 \pm 0.03$ &  $3.01 \pm 0.10$ &  $2.79 \pm 0.02$ & $2.71 \pm 0.01$ &  $2.74 \pm 0.01$ &  $2.72 \pm 0.01$ & $2.71 \pm 0.01$\\
         Ours&  $\bm{5.44\pm 0.04}$ & $\bm{5.40 \pm 0.04}$ & $\bm{5.40 \pm 0.03}$ & $\bm{2.69 \pm 0.01}$ & $\bm{2.68 \pm 0.01}$ & $\bm{2.67 \pm 0.01}$ &  $\bm{2.70 \pm 0.01}$ & $\bm{2.69 \pm 0.01}$ & $\bm{2.68 \pm 0.01}$ \\
         \bottomrule
    \end{tabular}
\end{table*}

\begin{algorithm}[t]
\SetAlgoLined
\KwIn{Neural network $\Phi$, dataset $\cal{D}$, Lipschitz bound $K$ of \eqref{eq:dynamics}, state space $\cal{X}$, radius of lattices $\delta$.} 

Discretize $\cal{X}$ into $N_l$ lattices $\cal{G}_1, ..., \cal{G}_{N_l}$, build a k-d tree of $\cal{D}$, and initialize the estimation error bound $\Delta = 0$\;

Calculate the Lipschitz bound $\gamma$ of $\Phi$\;

\For{$i = 1, 2, ..., N_l$}{
    Query the k-d tree for the data points inside $\cal{G}_i$. If none, return the top $q$ closest data points to the lattice center $\bar{x}_i$. The obtained data points are labeled as: $\cal{S}_i = \{ (x_1, y_1), ..., (x_q, y_q) \}$\;

    Initialize the estimation error on $\cal{G}_i$ by $e_i = \infty$\;
    
    \For{$j = 1, 2, ..., |\cal{S}_i|$}{
        $\epsilon_j = \lVert \Phi(x_j) - y_j \rVert_2 + (K+\gamma) \max_{k} \lVert x_j - v_k \rVert_2 $ where $v_k$ are the vertices of $\cal{G}_i$\;

        Update $e_i$ by $e_i = \min(e_i, \epsilon_j)$\;
    }

    Update $\Delta$ by $\Delta = \max(\Delta, e_i)$\;
}

Return $\Delta$ as the bound on estimation error\;
\caption{Bounding the estimation error}
\label{algo:est_err}
\end{algorithm}

\begin{figure}
    \centering
    \includegraphics[trim={5 10 5 10}, width=0.55\linewidth]{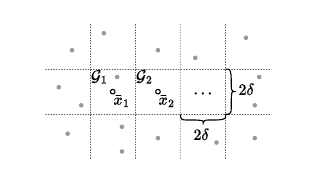}
    \caption{Discretization of the state space $\cal{X}$ in 2D. Solid gray circles ($\bullet$) represent data points in the dataset $\cal{D}$. Empty circles ($\circ$) represent lattice centers. }
    \label{fig:lattices}
\end{figure}

A practical way of finding the maximal estimation error is presented in Algorithm~\ref{algo:est_err}. The state space $\cal{X}$ is discretized into 
$N_l$ lattices \cite{xiang2018output}, and each lattice is a $\ell_\infty$ ball 
\begin{equation}
    \cal{G}_i = \{ x \in \R^n \mid \lVert x - \bar{x}_i  \rVert_\infty \leq \delta \}
\end{equation}
where $\bar{x}_i$ is the center of this ball (see Fig.~\ref{fig:lattices}). $\bar{x}_i$ are generated such that $\cal{X} \subseteq \union_{i=1,..., N_l} \cal{G}_i$. We first bound the estimation error for each lattice. The estimation error over $\cal{X}$ can then be bounded by the maximum of the bounds among all the lattices. For practical purposes, we take the labels $y_i$ as the ground truth in Algorithm~\ref{algo:est_err}, because we do not have exact knowledge of $f$. As the query of a k-d tree has complexity of $O(\log(|\cal{D}|))$, the complexity of Algorithm~\ref{algo:est_err} is $O(\log(|\cal{D}|) N_l)$.

Finally, let $z(t)$ be the solution to 
\begin{equation}\label{eq:est_dynamics}
    \dot{z} = \Phi(z),
\end{equation}
and the following result gives us a bound on $\lVert x(t) - z(t) \rVert_2 $ when \eqref{eq:dynamics} and \eqref{eq:est_dynamics} share the same initial conditions. 

\begin{proposition}
    Let $\cal{X}$ be compact, and $x(t)$ and $z(t)$ be the solutions to \eqref{eq:dynamics} and \eqref{eq:est_dynamics} with $x_0 = z_0$, respectively. If $\max_{x \in \cal{X}} \lVert \Phi(x) - f(x) \rVert_2 \leq a $, then we have $\lVert x(t) - z(t) \rVert_2 \leq \frac{a}{\gamma} (e^{\gamma t} -1)$.
\end{proposition}
\begin{proof}
    Let $d(t) = \lVert x(t) - z(t) \rVert_2$. Differentiate $d(t)^2$ w.r.t. $t$, and we have
    $
        2 d(t) \dot{d}(t) = 2 (\dot{x}(t) - \dot{z}(t))^\top (x(t) - z(t))
        = 2 (f(x(t)) - \Phi(z(t)))^\top (x(t) - z(t))
        \leq 2 \lVert f(x(t)) - \Phi(z(t)) \rVert_2 d(t).
    $
    As $d(t) \geq 0$, $\dot{d}(t) \leq \lVert f(x(t)) - \Phi(z(t)) \rVert_2 $. Since $\Lip(\Phi) \leq \gamma$ and $\max_{x \in \cal{X}} \lVert \Phi(x) - f(x) \rVert_2 \leq a $, we have $
        \dot{d}(t) \leq \lVert f(x(t)) - \Phi(x(t)) \rVert_2 + \lVert \Phi(x(t)) - \Phi(z(t)) \rVert_2 \leq a + \gamma d(t).
    $
    Finally, as $d(0) = 0$, it follows that $d(t) \leq \frac{a}{\gamma} (e^{\gamma t} -1)$ by the comparison lemma \cite[Lemma B.2]{khalil2015nonlinear}. 
\end{proof}

\section{Simulation Studies}
\begin{table}[t]
    \centering
    \caption{Structures of FCNs, LRNs, Lipschitz neural networks (ours) used for the simulation studies. $\operatorname{ReLU}(z) = \max(0, z)$. $\operatorname{LeakyReLU}(z) = \max(0, z) + 0.01 \min(0, z)$.}
    \label{tab:nn_structures}
    \begin{tabular}{ccc}
        \toprule
         &  Layer dimensions&  Activation function\\
         \midrule
         FCNs&  [64,64,64,64,64,64,64,2]&  ReLU\\
         LRNs&  [64,64,64,64,64,64,64,2]&  LeakyReLU\\
         Ours&  [64,64,64,64,64,64,64,2]&  ReLU\\
        \bottomrule
    \end{tabular}
\end{table}

In this section, we demonstrate the effectiveness of our approach proposed in Section \ref{sec:sys_id_through_lip_nn} on three simulated dynamical systems: a linear system, the Van der Pol oscillator, and a two-link planar robotic arm.

The observations of the system are the trajectories of $x(t)$ sampled at 100~\si{Hz} with a certain level of measurement noise. We perform low-pass filtering on the collected trajectories of $x(t)$ and approximate the value of $\dot{x}(t)$ by $\widehat{\dot{x}}(t)$, which is obtained from the fourth-order central difference method. As described in Section \ref{sec:sys_id}, we form the dataset $\cal{D}$ and randomly split\footnote{For comparison studies, the test-train split and downsampling of $\cal{D}_{\mathrm{train}}$ is the same across different trials.} it into a training set $\cal{D}_{\mathrm{train}}$ (80\%) and a test set $\cal{D}_{\mathrm{test}}$ (20\%). We take the test MSE (the MSE calculated on the test set $\cal{D}_{\mathrm{test}}$) as the figure of merit and compare our method with classic fully connected networks (FCNs) and Lipschitz regularized networks (LRNs) in \cite{negrini2021system, negrini2023robust}. In contrast with our approach, the authors of \cite{negrini2021system} add a Lipschitz regularization term $\beta \widehat{\Lip}(\Psi)$ (where $\beta>0$ is the weight and $\widehat{\Lip}(\Psi)$ is the Lipschitz constant of the LRN $\Psi$ estimated on the training batch) to the MSE loss during the training of the neural networks, making it a suitable comparison. 

For FCNs, the main hyperparameter is the weight decay coefficient (equivalent to $\ell_2$ regularization). We repeat Algorithm~\ref{algo:sys_id} and report the best test MSE by varying the weight decay coefficient in $\{10^{-8},\allowbreak 10^{-7},\allowbreak 10^{-6},\allowbreak 10^{-5},\allowbreak 10^{-4},\allowbreak 10^{-3},\allowbreak 10^{-2},\allowbreak 10^{-1}\}$. For LRNs\footnote{Implementation based on their \href{https://github.com/enegrini/System-identification-through-Lipschitz-regularized-neural-networks.git}{code}.}, the main hyperparameter is the weight of the Lipschitz regularization term $\beta$. Similarly, we report the best test MSE by varying this weight in $\{10^{-8},\allowbreak 10^{-7},\allowbreak 10^{-6},\allowbreak 10^{-5},\allowbreak 10^{-4},\allowbreak 10^{-3},\allowbreak 10^{-2},\allowbreak 10^{-1} \}$. 

In each example, we report the Lipschitz bound and estimation error bound (on $\cal{X}$) (Table~\ref{tab:max_error}), the best test MSE for the three methods (Table~\ref{tab:test_mse}), and the structure of the neural networks (Table~\ref{tab:nn_structures}). In Table~\ref{tab:max_error}, the Lipschitz bound of FCNs and LRNs are calculated after the training using LipSDP proposed in \cite{fazlyab2019efficient}, while that of Lipschitz neural networks is fixed and calculated according to Proposition~\ref{prop:lip_bound}. After determining the Lipschitz bound, we use Algorithm~\ref{algo:est_err} to calculate the bound on the estimation error. In Table~\ref{tab:nn_structures}, LeakyReLU is used for LRNs according to \cite{negrini2021system}. Table~\ref{tab:max_error} shows our method achieves the best theoretical bound.

\begin{figure*}[t]
    \centering
    \begin{subfigure}[b]{0.23\textwidth}
        \centering
        \includegraphics[width=\linewidth]{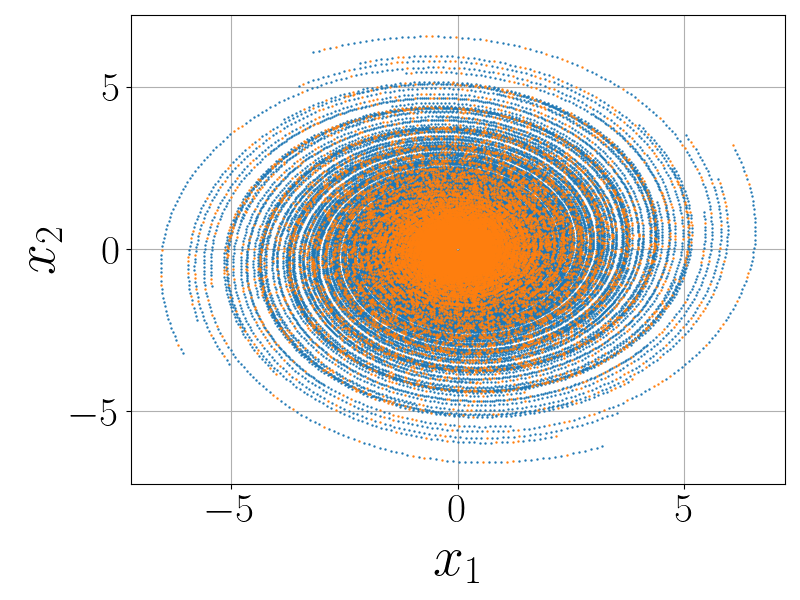}
        \caption{Linear system.}
        \label{fig:eg1_data}
    \end{subfigure}
    \hspace{1pt}
    \begin{subfigure}[b]{0.23\textwidth}
        \centering
        \includegraphics[width=\linewidth]{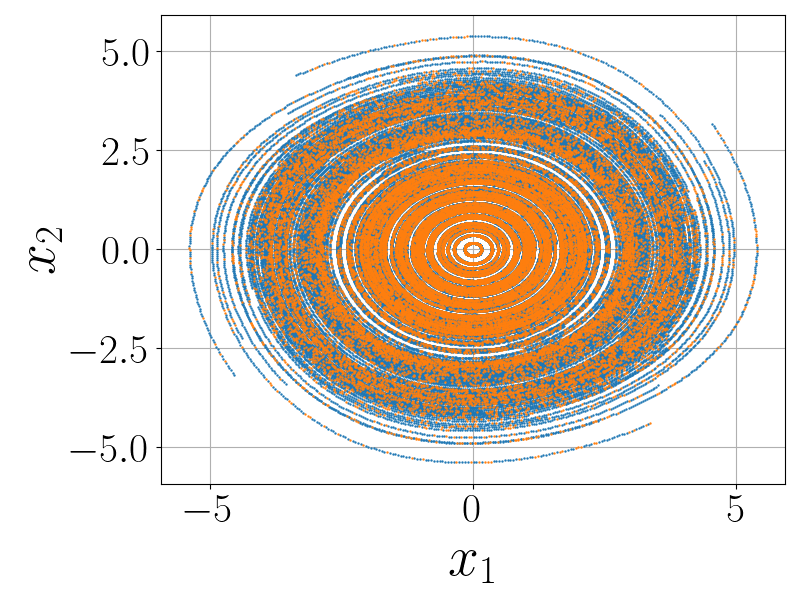}
        \caption{Van der Pol oscillator.}
        \label{fig:eg2_data}
    \end{subfigure}
    \hspace{1pt}
    \begin{subfigure}[b]{0.5\textwidth}
        \centering
        \includegraphics[width=0.46\linewidth]{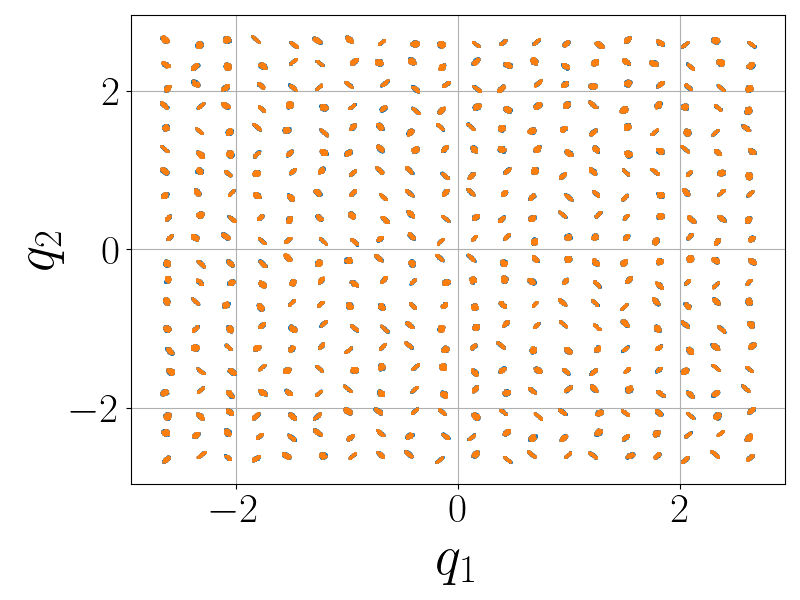}
        \hspace{1 pt}
        \includegraphics[width=0.46\linewidth]{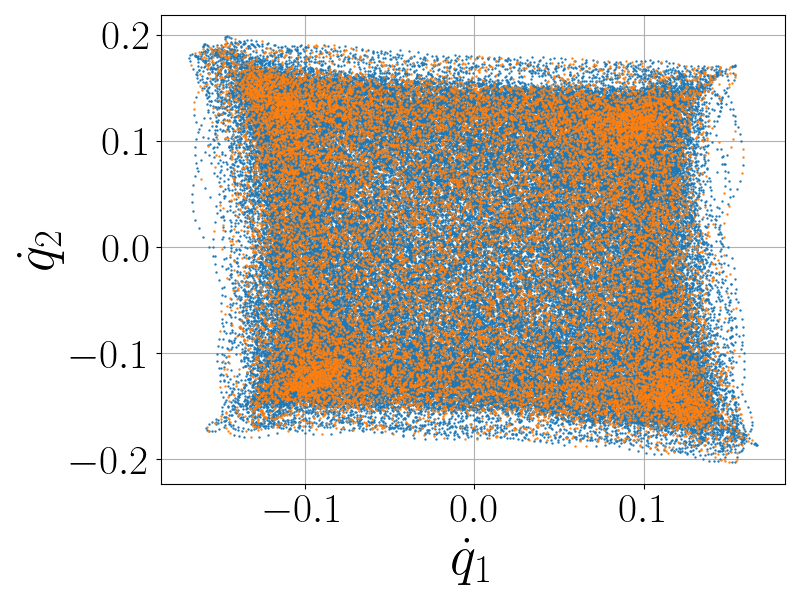}
        \caption{Two-link robotic arm. Left: $(q_1, q_2)$ plane. Right: $(\dot{q}_1, \dot{q}_2)$ plane.}
        \label{fig:eg3_data}
    \end{subfigure}
    \caption{Visualization of training (in blue) and test (in orange) data for the simulation studies.}
    \label{fig:data}
\end{figure*}

\begin{figure*}[t]
    \centering
    \begin{subfigure}[b]{0.27\textwidth}
        \centering
        \includegraphics[width=\linewidth]{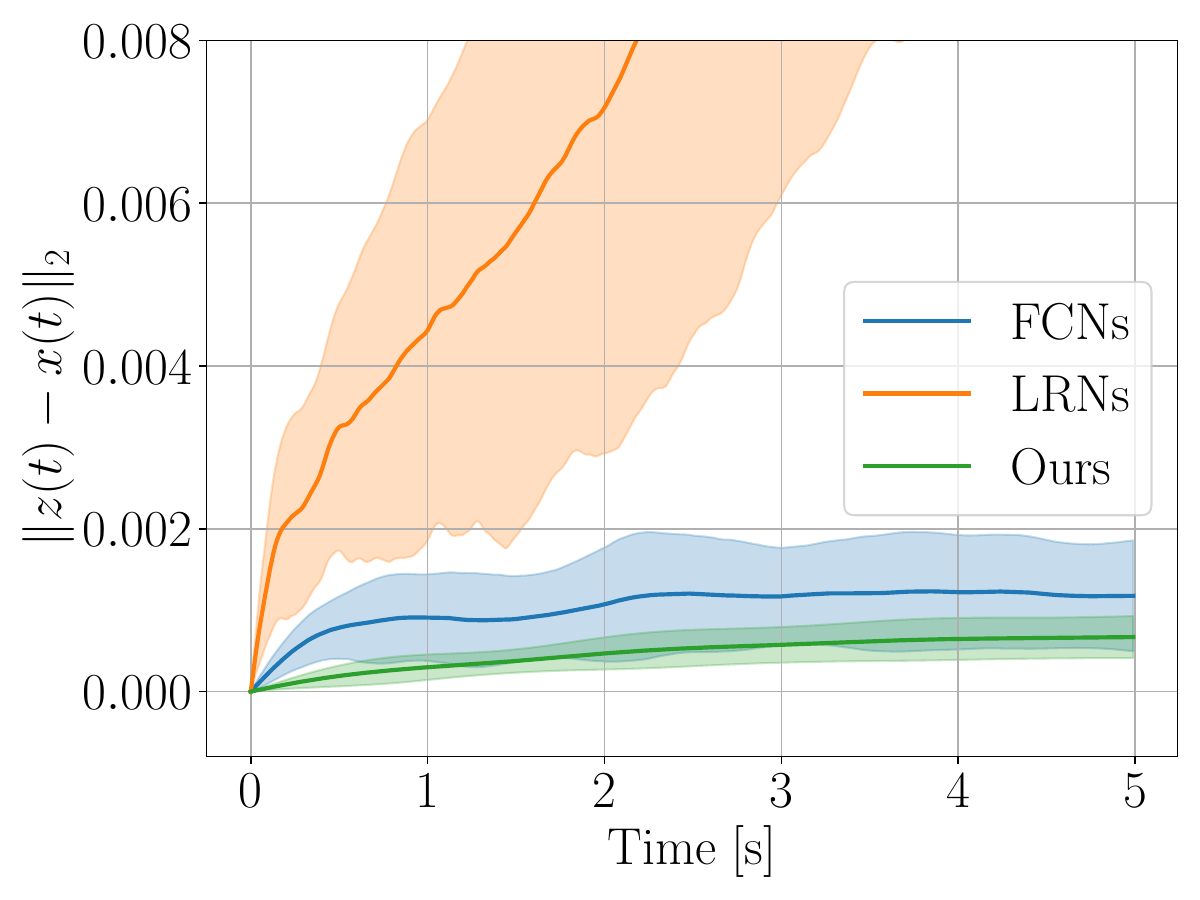}
        \caption{Linear system.}
        \label{fig:eg1_traj_error}
    \end{subfigure}
    \hspace{1pt}
    \begin{subfigure}[b]{0.27\textwidth}
        \centering
        \includegraphics[width=\linewidth]{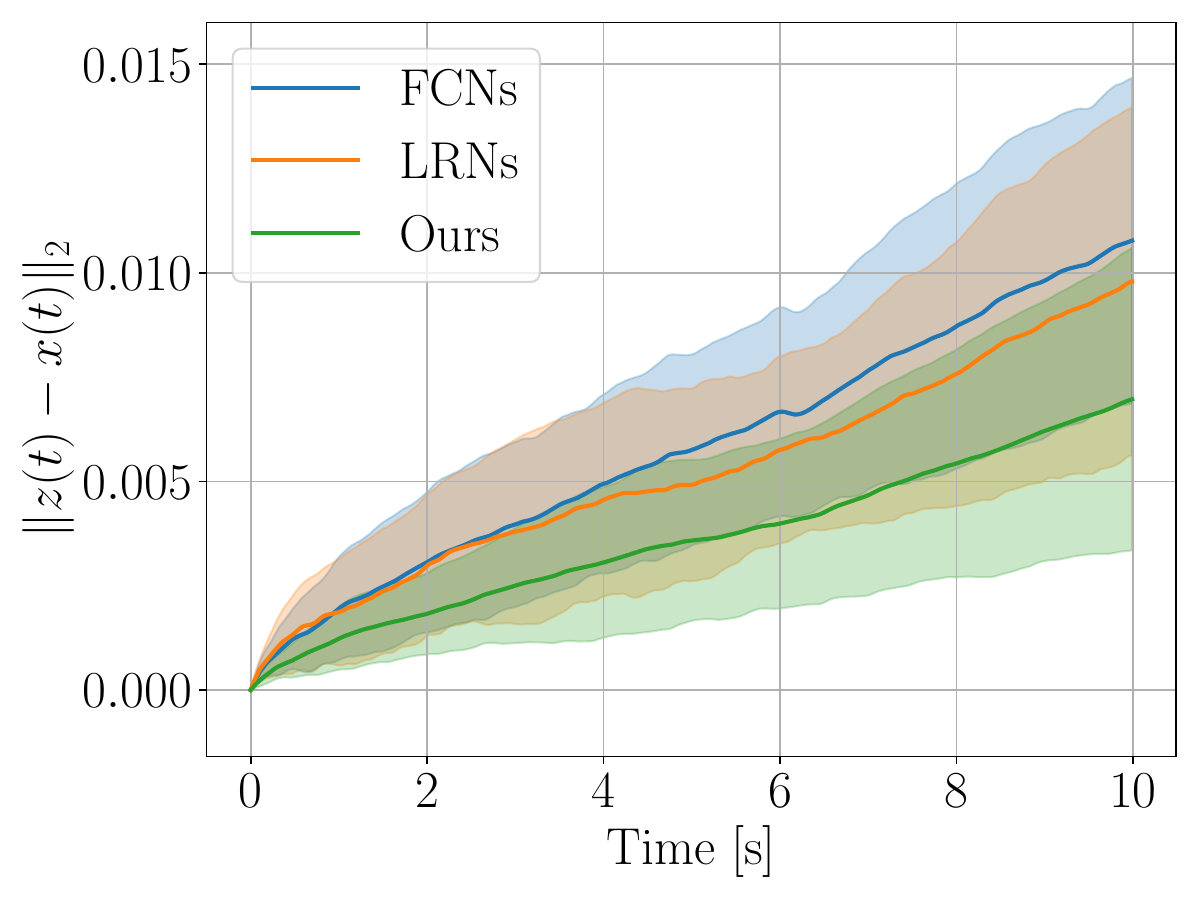}
        \caption{Van der Pol oscillator.}
        \label{fig:eg2_traj_error}
    \end{subfigure}
    \hspace{1pt}
    \begin{subfigure}[b]{0.27\textwidth}
        \centering
        \includegraphics[width=\linewidth]{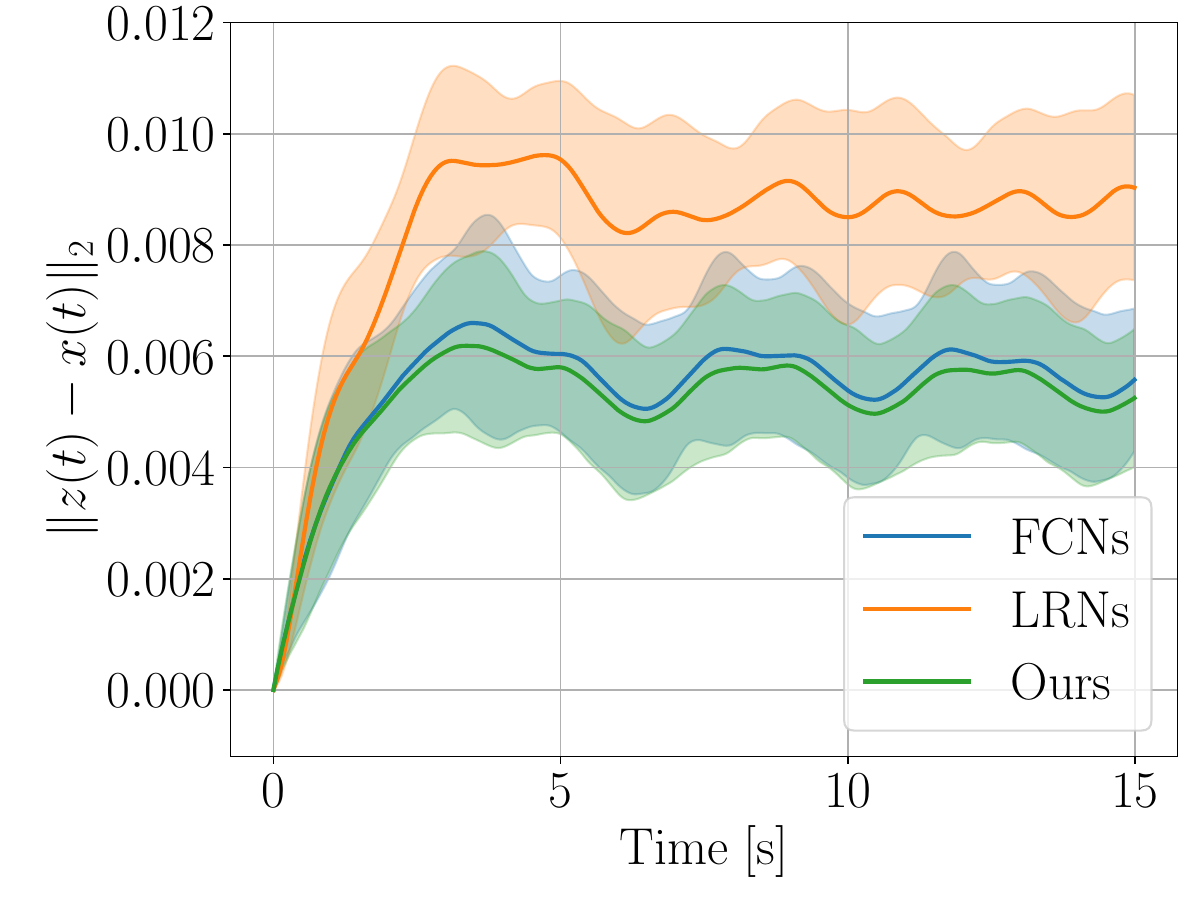}
        \caption{Two-link robotic arm.}
        \label{fig:eg3_traj_error}
    \end{subfigure}
    \caption{Mean and standard variation of the $\ell_2$ error of the simulated trajectories (using 100\% training data).}
    \label{fig:traj_error}
\end{figure*}

\subsection{Linear System}
\begin{figure}[t]
    \centering
    \includegraphics[width=0.45\textwidth]{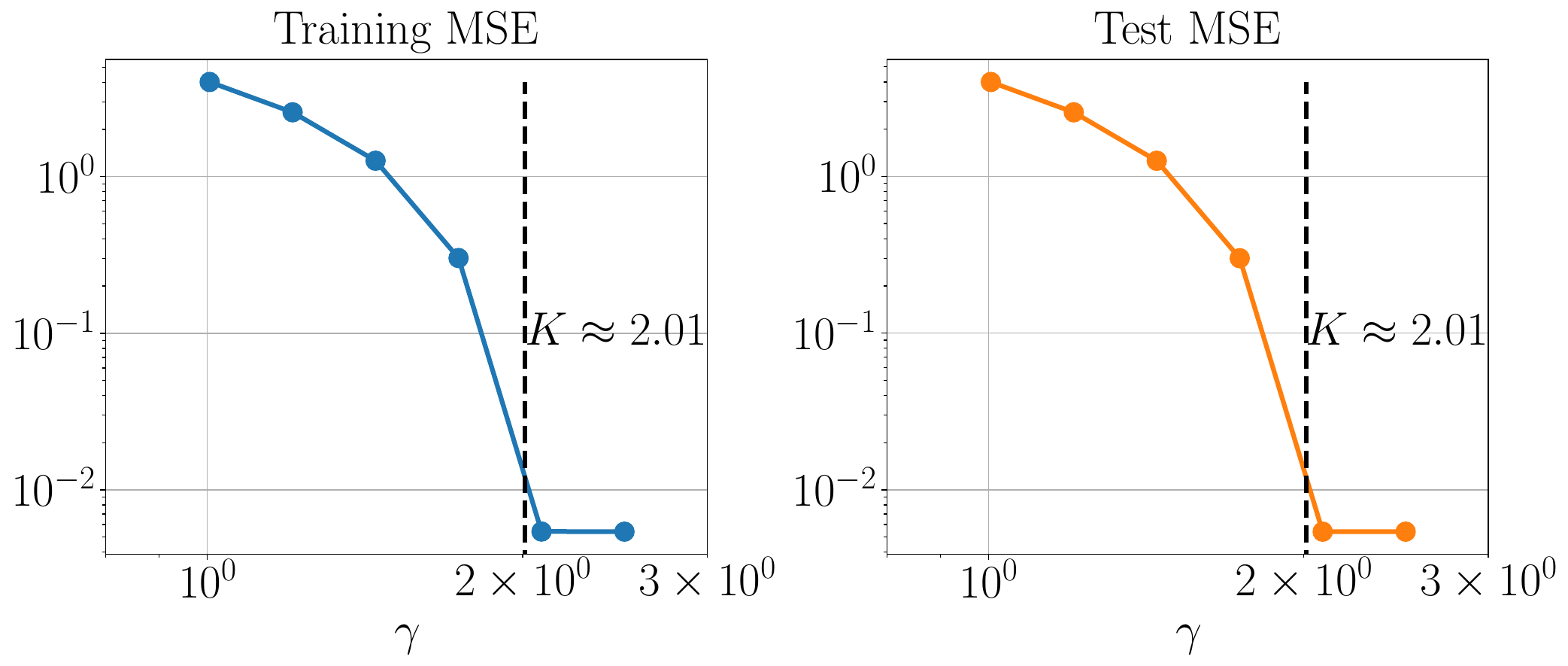}
    \caption{Training (left) and test (right) MSEs of Lipschitz neural networks w.r.t. the Lipschitz bound (using 100\% training data).}
    \label{fig:eg1_expressivity}
\end{figure}

We first test our approach on the following linear system
\begin{equation}\label{eq:linear_sys}
    \dot{x}_1 = - 0.2 x_1 + 2 x_2, \quad \dot{x}_2 = -2 x_1 - 0.2 x_2.
\end{equation}
This is a stable system, and the poles are $-0.2 \pm 2 i$. We are interested in fitting the dynamic over $\cal{X} = \{ (x_1, x_2) \mid -3 \leq x_1, x_2 \leq 3 \}$. We sample 100 trajectories of 12 seconds, resulting in a total of 120k data points (see Fig.~\ref{fig:eg1_data}). The outputs are $y_1 = x_1 + \omega_1$ and $y_2 = x_2 + \omega_2$ where $\omega_1, \omega_2 \sim \cal{N}(0,\sigma^2)$ with $\sigma^2 = 10^{-4}$. In Table~\ref{tab:test_mse}, we see that our method has lower test MSE than FCNs and LRNs in general, and the advantage is more evident when using only 25\% of the training data. In Table~\ref{tab:max_error}, our method has both the smallest Lipschitz bound (89.37\% smaller than FCNs and 85.84\% than LRNs) and the smallest estimation error bound (77.06\% smaller than FCNs and 70.78\% than LRNs when $\delta = 0.05$, and 78.39\% than FCNs and 72.43\% than LRNs when $\delta = 0.025$). Table~\ref{tab:max_error} also shows that soft regularization can result in a large variation of the Lipschitz bounds for both FCNs and LRNs. In addition, we simulate 100 trajectories using the learned dynamics (models with the best test MSE), with initial conditions uniformly sampled from $\cal{X}$. We then visualize the mean and standard deviation of the $\ell_2$ error relative to the trajectories generated by the true dynamics in Fig.~\ref{fig:eg1_traj_error}. We see that our method achieves the smallest error with the lowest variance among the three methods.

The Lipschitz bound of the RHS of \eqref{eq:linear_sys} is estimated to be $K \approx 2.01$ using finite difference on the collected data (assuming that the sampled points are sufficiently dense). Figure~\ref{fig:eg1_expressivity} demonstrates an underfitting behavior when $\gamma$, the Lipschitz bound of $\Phi$, is much smaller than $K$. This suggests that choosing a $\gamma$ that is at least comparable with $K$ is an effective strategy for enhancing performance.  

We remark that LRNs are more significantly affected by the amount of training data than both FCNs and our method. During the training of LRNs, the practical estimation of the Lipschitz constant is carried out among the training data, which is typically local and, most importantly, dependent on the amount of training data. Therefore, LRNs exhibit a large increase in test MSE when the training set is small.  

\subsection{Van der Pol Oscillator}
The Van der Pol oscillator, with dynamics
\begin{equation}\label{eq:van_der_pol}
\dot{x}_1 = x_2, \quad \dot{x}_2 = \mu (1 - x_1^2) x_2 - x_1,
\end{equation}
and $\mu >0$, is used as a system identification example in \cite{egan2024automatically}. For this example, $\mu = 0.02$ and $\cal{X} = \{ (x_1, x_2) \mid -2.5 \leq x_1, x_2 \leq 2.5 \}$. We sample 400 trajectories of 5 seconds, leading to a total of 200k data points (Fig.~\ref{fig:eg2_data}). The outputs are $y_1 = x_1 + \omega_1$ and $y_2 = x_2 + \omega_2$ where $\omega_1, \omega_2 \sim \cal{N}(0,\sigma^2)$ with $\sigma^2 = 5 \times 10^{-5}$. The Lipschitz bound of the RHS of \eqref{eq:van_der_pol} is estimated to be $K \approx 1.65$ on the collected data. Table~\ref{tab:test_mse} shows that our method has the lowest test MSEs for all three percentages of the training data. In Table~\ref{tab:max_error}, our method again has the smallest Lipschitz bound (23.72\% smaller than FCNs and 76.57\% than LRNs) and estimation error bound (17.86\% smaller than FCNs and 69.33\% than LRNs when $\delta = 0.05$, and 18.18\% than FCNs and 69.39\% than LRNs when $\delta = 0.025$). We again simulate 100 trajectories using the learned dynamics (models with the best test MSE), with initial conditions uniformly sampled from $\cal{X}$. Fig.~\ref{fig:eg2_traj_error} shows that our method achieves the smallest error compared with FCNs and LRNs.

\subsection{Two-Link Robotic Arm}
\begin{figure}[t]
    \centering
    \includegraphics[width=0.30\textwidth]{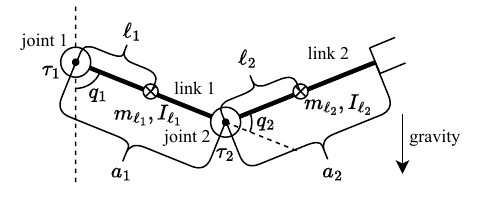}
    \caption{Two-link planar arm.}
    \label{fig:eg3_arm}
\end{figure}
Our method can be applied to learn the uncertainties in a controlled system. Consider the two-link robotic planar arm shown in Fig.~\ref{fig:eg3_arm}. $m_{\ell_i}$, $a_i$, and $I_{\ell_i}$ are the mass, length, and moment of inertia of link $i$, respectively. $\ell_i$ is the distance between the axis of joint $i$ and the center of mass of link $i$. $m_{m_i}$, $k_{r_i}$, and $I_{m_i}$ are the mass, gear reduction ratio, and moment of inertia of motor $i$ (which is located on the axis of joint $i$), respectively. $q_i$ and $\tau_i$ are the angle and the input torque of joint $i$, respectively. The dynamics of the arm are
\begin{equation}\label{eq:two_link_arm}
    \frac{dq}{dt} = \dot{q}, \quad \frac{d \dot{q}}{dt}  = M^{-1}(q) (\tau - C(q,\dot{q})\dot{q} - F_f(\dot{q}) - g(q))
\end{equation}
where $q = [q_1, q_2]^\top$, $\tau = [\tau_1, \tau_2]^\top$, $M(q) \in \R^{2 \times 2}$ is the inertia matrix, $C(q,\dot{q})\dot{q}$ captures the centrifugal and Coriolis effects,  $g(q) \in \R^2$ are the moments generated by the gravity, and $F_f(\dot{q})$ represents the unknown friction torques. Due to page limit, the exact expressions for $M(q)$, $C(q,\dot{q})$, and $g(q)$ are omitted and can be found in \cite[Section 7.3.2]{siciliano2008robotics}. In this example, we consider $F_f(\dot{q}) = F_v \dot{q} + F_c \tanh (s_c \dot{q})$ where $F_v \dot{q}$ are the viscous friction torques and $F_c \tanh (s_c \dot{q}) \approx F_c \sign (\dot{q})$ approximates the Coulomb friction torques. The parameter values used for this study are: $a_1 = a_2 = 0.8$~\si{m}, $\ell_1 = \ell_2 = 0.4$~\si{m}, $m_{\ell_1} = m_{\ell_2}=20$~\si{kg}, $I_{\ell_1} = I_{\ell_2} = 5$~\si{kg . m^2}, $k_{r 1} = k_{r 2}=100$, $m_{m_1} = m_{m_2} = 2$~\si{kg}, $I_{m_1} = I_{m_2} = 0.01$ ~\si{kg . m^2}, $F_v = \diag([40,40])$ in \si{kg . m^2 . s^{-1}} units, $F_c = \diag([2,2])$ in \si{kg . m^2. s^{-1}} units, and $s_c = 10$.

The task is to learn the term $M^{-1}(q) F_f(\dot{q})$ from data. We sample 400 trajectories of 3 seconds, generating a total of 120k data points (see Fig. \ref{fig:eg3_data}). The state space is $\cal{X} = \{ (q_1, q_2, \dot{q_1}, \dot{q}_2) \mid |q_1|, |q_2| \leq 3\pi/4, |\dot{q_1}|, |\dot{q}_2| \leq 0.1\}$. In each trajectory, the robot arm starts from initial joint angles $q_0$ and is controlled by $\tau = g(q)+C(q,\dot{q})\dot{q}+M(q)[-K_p (q - q_0) - K_d \dot{q}] + \epsilon (t)$ where $K_p = I$, $K_d = 2I$, and $\epsilon (t) = [100\sin(2\pi t + \varphi_1), 100\sin(2\pi t + \varphi_2)]^\top$ (with the phases $\varphi_1$ and $\varphi_2$ randomly generated for each trajectory). We collect the control inputs $\tau$ and the outputs $y = [q^\top, \dot{q}^\top]^\top + \omega$ where $\omega \sim \cal{N}(0, \sigma^2 I)$ and $\sigma^2 = 5 \times 10^{-5}$. The Lipschitz bound of the term $M^{-1}(q) F_f(\dot{q})$ is estimated to be $K \approx 0.59$ on the collected data. 

As in Table~\ref{tab:test_mse}, our method has the best test MSEs when using 25\% and 50\% training data and has the same test MSE as FCNs when using 100\% training data. In Table~\ref{tab:max_error}, our method again has the smallest Lipschitz bound (83.20\% smaller than FCNs and 76.58\% than LRNs) and estimation error bound (79.79\% smaller than FCNs and 72.30\% than LRNs when $\delta = 0.05$, and 79.81\% than FCNs and 72.33\% than LRNs when $\delta = 0.025$). We simulate 100 trajectories using the learned dynamics (with initial joint angles uniformly sampled from $\cal{X}$ and zero joint velocity), and $\tau = g(q)+C(q,\dot{q})\dot{q}+M(q)[-K_p (q - q_0) - K_d \dot{q}] + \epsilon (t)$ where $K_p = I$, $K_d = 2I$, and $\epsilon (t) = [30\sin(0.5 \pi t + \varphi_1), 30\sin(0.5 \pi t + \varphi_2)]^\top$ (with $\varphi_1$ and $\varphi_2$ randomly generated for each trajectory). Fig.~\ref{fig:eg2_traj_error} shows that our method achieves the smallest error compared with FCNs and LRNs.

\section{Conclusion}
In this work, we propose a novel system identification method using Lipschitz neural networks. Our theoretical analysis provides an approximation error bound and a bound on the difference between the simulated solution by our method and the solution to the true system under mild assumptions. Comparison studies demonstrate the potential of our method in enhancing the precision and reliability of system identification using neural networks.

\bibliographystyle{IEEEtran}
\bibliography{master}

\end{document}